\documentclass{article}

     \PassOptionsToPackage{numbers, compress}{natbib}

     \usepackage[preprint]{neurips_2019}

\usepackage[utf8]{inputenc} 
\usepackage[T1]{fontenc}    
\usepackage{hyperref}       
\usepackage{url}            
\usepackage{booktabs}       
\usepackage{amsfonts}       
\usepackage{nicefrac}       
\usepackage{microtype}      
\usepackage{amsmath}
\usepackage{amsthm}
\usepackage{amssymb}
\usepackage{algorithm}
\usepackage{algorithmic}
\usepackage{graphicx}
\usepackage{subfigure}
\usepackage[small]{caption}
\usepackage{enumitem}
\usepackage{multirow}
\usepackage{wrapfig}
\usepackage{lipsum}

\usepackage{appendix}
\setlist[itemize]{itemsep=-0.25ex,leftmargin=2.5ex}
\setlist[enumerate]{itemsep=-0.25ex,leftmargin=2.5ex}

\def\B#1{\boldsymbol #1}

\newtheorem{lemma}{Lemma}
\newtheorem{theorem}{Theorem}

\title{Bandit Learning for Diversified Interactive Recommendation}

\author{%
  Yong Liu \\
  Nanyang Technological University \\
  \texttt{stephenliu@ntu.edu.sg} \\
  \And
  Yingtai Xiao \\
  University of Science and Technology of China \\
  \texttt{gzxyt@mail.ustc.edu.cn} \\
  \AND
  Qiong Wu\\
  Nanyang Technological University \\
  \texttt{wu.qiong@ntu.edu.sg} \\
  \And
  Chunyan Miao \\
  Nanyang Technological University \\
  \texttt{ascymiao@ntu.edu.sg} \\
  \And
  Juyong Zhang\\
  University of Science and Technology of China \\
  \texttt{juyong@ustc.edu.cn}
}

\begin{document}

\maketitle

\begin{abstract}
Interactive recommender systems that enable the interactions between users and the recommender system have attracted increasing research attentions. Previous methods mainly focus on optimizing recommendation accuracy. However, they usually ignore the diversity of the recommendation results, thus usually results in unsatisfying user experiences. In this paper, we propose a novel diversified recommendation model, named \underline{D}iversified \underline{C}ontextual \underline{C}ombinatorial \underline{B}andit (DC$^2$B), for interactive recommendation with users' implicit feedback. Specifically, DC$^2$B employs determinantal point process in the recommendation procedure to promote diversity of the recommendation results. To learn the model parameters, a Thompson sampling-type algorithm based on variational Bayesian inference is proposed. In addition, theoretical regret analysis is also provided to guarantee the performance of DC$^2$B. Extensive experiments on real datasets are performed to demonstrate the effectiveness of the proposed method. 
\end{abstract}

\section{Introduction}
Conventional recommender systems are usually developed in non-interactive manner and learn the user preferences from logged user behavior data. One main drawback of these recommender systems is that they cannot capture the changes of users' preferences in time. This requires the development of interactive recommender system that enables interactions between users and the recommender system~\cite{steck2015interactive}. In the literature, contextual bandit learning has been demonstrated to be a promising solution to interactive recommendation problems~\cite{li2010contextual,zhao2013interactive,tang2015personalized,wang2017factorization,qi2018bandit}. In these methods, the recommender system sequentially recommends a set of items to a user and adopts the user's immediate feedback to improve its recommendation policy.

In practice, users' implicit feedback (e.g., clicking history) are usually utilized to build recommender systems, because implicit feedback is user centric, and can be easily collected~\cite{shi2014collaborative,liu2014exploiting,liu2015boosting,liu2018dynamic}. However, the implicit feedback usually brings bias signals which make the recommendation problems much more challenging. This bias comes from the fact that the implicit feedback can only capture the positive user preferences (i.e., observed user-item interactions), and all negative user preferences are missing. Although the non-interaction between the user and an item is usually treated as negative user preference in previous research work~\cite{shi2014collaborative}, it does not explicitly indicate that the user dislikes the item, as non-interaction may also be caused by that the item has not been exposed to the user~\cite{liang2016modeling}.

In addition, previous interactive recommendation methods mainly focus on optimizing recommendation accuracy. They usually ignore other important properties of the recommendation results, for example the diversity of the recommended item set~\cite{kunaver2017diversity}. Therefore, the items in the recommendation lists generated by these approaches may usually be very similar with each other, and the recommendation results may only cover a small fraction of items. This usually leads to inferior user experiences, and thus reduces the commercial values of recommender systems.

In this paper, we propose a novel bandit learning framework for interactive recommender systems based on users' implicit feedback, which strives to achieve a good balance between accuracy and diversity in the recommendation results. To solve the bias problems caused by implicit feedback, we model the interactions between users and the recommender system from two perspectives: i) \textbf{Diversified Item Exposure}: the recommender system selects a set of relevant yet diverse items to expose to the user; ii) \textbf{User Engagements}: the user eventually engages with some of the exposed items (e.g., clicks on the items). Specifically, the determinantal point process (DPP)~\cite{kulesza2012determinantal} is employed to select a set of diverse items to expose to users, considering both the qualities of items and the diversity of the selected item set. The advantage of DPP is that it explicitly models the probability that an item set would be selected to show to the user, thus can help solve the bias problem caused by implicit feedback~\cite{liang2016modeling}. In addition, the contextual features of items are also utilized to model the observed user engagements on the recommended items.

To summarize, the major contributions made in this paper are as follows: (1) we propose a novel bandit learning method, i.e., \underline{D}iversified \underline{C}ontextual \underline{C}ombinatorial \underline{B}andit (DC$^{2}$B), to improve the recommendation diversity for interactive recommender systems; (2) we propose a variational Bayesian inference algorithm under the Thompson sampling framework to learn the model parameters; (3) we also provide theoretical regret analysis for the proposed DC$^{2}$B method; 
(4) we perform extensive experiments on real datasets to demonstrate the effectiveness of DC$^{2}$B in balancing the recommendation accuracy and diversity, especially on large and sparse datasets.

\section{Related Work}
\label{sec:relatedwork}
\textbf{Diversified Recommendation.} \noindent 
One major group of diversified recommendation methods are based on greedy heuristics. The pioneering work is maximal marginal relevance (MRR)~\cite{carbonell1998use}, which defines a marginal relevance to combine the relevance and diversity metrics, and creates a diversified ranking of items by choosing an item in each interaction such that it maximizes the marginal relevance. Other greedy heuristics methods vary in the definition of the marginal relevance, often in the form of a sub-modular objective function~\cite{qin2013promoting,sha2016framework}, which can be solved greedily with an approximation to the optimal solution.
Another group of methods are based on refinement heuristics, which usually re-rank a pre-ranked item list through post-processing actions~\cite{zhang2008avoiding,antikacioglu2017post}.
From another perspective, \cite{cheng2017learning} formulates the diversified recommendation problem as a supervised learning task, and proposes a diversified collaborative filtering model to solve the optimization problems. Recently, DPP has been demonstrated to be effective in modeling diversity in various machine learning problems~\cite{kulesza2012determinantal}, and some recent work~\cite{chen2018fast,wilhelm2018practical,wu2019adversarial} employs DPP to improve recommendation diversity. Overall, these diversified recommendation methods are developed for non-interactive recommender systems.

\textbf{Interactive Recommendation.} \noindent Contextual bandit has been often used for building interactive recommender systems. These methods mainly focus on optimizing the recommendation accuracy. For instance, \citep{li2010contextual} proposes a contextual bandit algorithm, named LinUCB, which sequentially recommended articles to users based on the contextual information of users and articles. \cite{zhao2013interactive} combines probabilistic matrix factorization with Thompson sampling and upper confidence bound based bandit algorithms to interactively select items. 
\cite{tang2015personalized} proposes a parameter-free bandit approach that uses online bootstrap to learn the online recommendation model. Recently, \cite{wang2017interactive} extends the LinUCB to incorporate users' social relationships into interactive recommender system. \cite{wang2017factorization} proposes a factorization-based bandit approach to solve the online interactive recommendation problem.
Moreover, in~\cite{qi2018bandit}, the Thompson sampling framework is employed to solve the bandit problems with implicit feedback, where the implicit feedback is modeled as a composition of user result examination and relevance judgement. There also exist some interactive recommender systems focus on promoting the recommendation diversity. For example, \cite{qin2014contextual} proposes a contextual combinatorial bandit framework, incorporating the entropy regularizer~\cite{qin2013promoting} to diversify the recommendation results.
Differing from~\cite{qin2014contextual}, DC$^2$B is a full Bayesian recommendation framework which is more effective in balancing the recommendation accuracy and diversity, especially on larger and sparser datasets.


\section{Problem Formulation}
We employ contextual bandit to build the diversified interactive recommender system. The recommender system is treated as an agent, and each item is treated as an arm.
Let $\mathcal{A}=\{a_{i}\}_{i=1}^{N}$ denote the set of $N$ arms (i.e., items). We assume each arm $a_{i}$ has a contextual feature vector $\B{x}_i \in \mathbb{R}^{1 \times d}$ summarizing its side information, and denote the features of all arms by $\B{X} \in \mathbb{R}^{N \times d}$. At each trial, the recommender agent would firstly choose a subset of arms $\mathcal{S}$ from $\mathcal{A}$, considering the qualities of the arms and the diversity of selected arms. In the literature, $\mathcal{S}$ is usually called as a \emph{super arm}. Here, we empirically define the quality of an arm $a_{i}$ as follows:
\begin{equation}
  r_{i} = \exp(\B{\theta} \B{x}_{i}^{\top}),
\end{equation}
where $\B{\theta}$ is the bandit parameter that describes the user preferences. The diversity of the selected super arm $\mathcal{S}$ can be measured by the intra-list distance metric~\cite{zhang2008avoiding}. Once a diversified super arm $\mathcal{S}$ has been selected according to a policy $\pi$ and displayed to the user, the user's engagements on the displayed items (e.g., clicks on the items) are used as the rewards for the recommender agent to optimize its recommendation policy. Through the interactions with the user, the recommender agent aims to adjust its super arm selection strategy to maximize its cumulative reward over time.

\textbf{Diversified Item Exposure:} \noindent The DPP is an elegant probabilistic model with the ability to model diversity in various machine learning problems~\cite{kulesza2012determinantal}. In this work, we utilize DPP to model the selection probability of a relevant yet diverse super arm $\mathcal{S}$. Formally, a DPP $\mathcal{P}$ on the set of candidate arms $\mathcal{A}$ is a probability measure on $2^{\mathcal{A}}$, describing the probability for the set of all subsets of $\mathcal{A}$. If $\mathcal{P}$ assigns nonzero probability on the empty set $\emptyset$, there exists a real, positive semi-definite kernel matrix $\B{L} \in \mathbb{R}^{N \times N}$, such that the probability of the super arm $\mathcal{S}$ can be defined as follows:
\begin{align}
p(\mathcal{S}) = \frac{\det(\B{L}
_{[\mathcal{S}]})}{\det(\B{L}+\B{I})},
\end{align}
where $\B{I}$ is the identity matrix, $\B{L}_{[\mathcal{S}]}\equiv [\B{L}_{ij}]_{a_{i}, a_{j} \in \mathcal{S}}$ is the sub-matrix of $\B{L}$.
As revealed in~\cite{kulesza2012determinantal}, $\B{L}$ can be written as a Gram matrix, $\B{L}=\B{V}\B{V}^\top$, where the rows of $\B{V}$ are vectors representing the arms. Following previous studies~\cite{chen2018fast,wilhelm2018practical}, we empirically set $\B{V}_{i}= (r_{i})^{\alpha}\B{x}_{i}$, where $\alpha > 0$ is a parameter controlling the impacts of item qualities. 
More details about DPP can be found in~\cite{kulesza2012determinantal}.

\textbf{User Engagements: }\noindent The user's engagements on displayed items are expressed by her implicit feedback (e.g., clicks on the items), which is usually described by a set of binary variables. If the user engages in the arm $a_{i}$, we set $y_{i}$ to 1; otherwise, we set $y_{i}$ to 0. Once an arm $a_{i} \in \mathcal{S}$ has been displayed to the user, we assume the user's engagements on $a_{i}$ is only determined by its quality. Thus, the probability of the observed user engagement on $a_{i}$, i.e., $y_{i}=1$, can be defined as follows:
\begin{equation}
  p_{i}
  \triangleq \rho(\B{\theta} \B{x}_{i}^{\top})
  =\frac{\exp(\B{\theta} \B{x}_{i}^{\top})}{1+\exp(\B{\theta} \B{x}_{i}^{\top})}
  =\frac{r_{i}}{1+r_{i}}.
\end{equation}
This can be explained as that when an arm $a_{i}$ is offered to the user, the user engages in this arm or a \emph{virtual arm} $a_{0}$ with a relevance score 1. 
Based on these assumptions, we can define the joint probability of observed user engagements $\mathcal{Y}=\{y_{i}|a_{i} \in \mathcal{S}\}$ as follows:
\begin{align}
  p(\mathcal{Y}, \mathcal{S}, \B{\theta})=p(\B{\theta})p(\mathcal{S}|\B{\theta})p(\mathcal{Y}|\mathcal{S}, \B{\theta})
  =p(\B{\theta})\frac{\det(\B{L}
_{[\mathcal{S}]})}{\det(\B{L}+\B{I})}\prod_{a_{i} \in \mathcal{S}}p_{i}^{y_{i}}(1-p_{i})^{1-y_{i}},
\end{align}
where $p(\B{\theta})$ is the prior assigned to bandit parameters. In addition, we assume $p(\B{\theta})$ follows a Gaussian distribution $\mathcal{N}(\B{m}, \B{\Sigma})$, and $\B{m}$, $\B{\Sigma}$ are bounded. This assumption is typically used in practice.

\begin{algorithm}[tb]
   \caption{Thompson sampling for DC$^2$B}
   \label{alg:TS}
\begin{algorithmic}
   \STATE Initialize $\B{m}=\B{0}$, $\B{\Sigma}=\lambda \B{I}$, and $\mathcal{R}=\emptyset$.
   \FOR{$t=0$ {\bfseries to} $T$}
   \STATE $\mathcal{A}_{t} \leftarrow \mathcal{A} - \mathcal{R}$, $\B{X}_{t}=\{\B{x}_i | a_{i}\in \mathcal{A}_t\}$
   \STATE Randomly sample $\widehat{\B{\theta}}\sim \mathcal{N}(\B{m},\B{\Sigma})$
   \STATE $\mathcal{S} \leftarrow \mathcal{O}(\widehat{\B{\theta}},\B{X}_t)$
   \STATE Play super arm $\mathcal{S}$ and observe the reward $\mathcal{Y}$
   \STATE Update $\B{\Sigma}$ and $\B{m}$ according to Eq.~\eqref{eq:sigmaupdate}, Eq.~\eqref{eq:meanupdate}, and Eq.~\eqref{eq:xiupdate}.
   \STATE $\mathcal{R} \leftarrow \mathcal{R} \cup \mathcal{S}$
   \ENDFOR
\end{algorithmic}
\end{algorithm}

\section{Parameter Inference}
\label{sec:algorithm}
Once a newly obtained observation $(\mathcal{S}, \mathcal{Y})$ is available, we employ variational Bayesian inference~\cite{blei2017variational} to develop a closed form approximation to the posterior of $\B{\theta}$. According to~\cite{blei2017variational}, the approximated posterior $q(\B{\theta})$ of $\B{\theta}$ can be expressed as $\log q^{\ast}(\B{\theta})=\mathbb{E}_{param\neq \B{\theta}}[\log p(\mathcal{Y}, \mathcal{S}, \B{\theta})] +\mbox{const}.$ Moreover, based on the knowledge in Linear Algebra, we have $\det(\B{L}_{[\mathcal{S}]})= \prod_{a_{i}\in \mathcal{S}} r_{i}^{2\alpha} \det(\B{X}_{[\mathcal{S}]}\B{X}_{[\mathcal{S}]}^{\top})$ and $\det(\B{L}+\B{I})=\exp(\mbox{tr}(\log(\B{L}+\B{I}))$. Then, we can have the following log-likelihood function:
\begin{align}
\label{eq:likelihood}
   \log p(\mathcal{Y}, \mathcal{S}|\B{\theta})
   = \sum_{a_{i}\in \mathcal{S}}(\varphi(y_{i}, p_{i})+ 2\alpha \log r_{i}) + \log \det(\B{X}_{[\mathcal{S}]}\B{X}_{[\mathcal{S}]}^{\top}) -\sum_{j=1}^{N}\log (1+r_{j}^{2\alpha} \B{x}_j\B{x}_j^{\top}),
\end{align}
where $\varphi(y_{i}, p_{i})=y_{i}\log p_{i} +(1-y_{i}) \log (1-p_{i})$. In Eq.~\eqref{eq:likelihood}, the likelihood function is a logistic function, which is not conjugate with the Gaussian priors on $\B{\theta}$. To address this issue, the following Gaussian lower bound on the logistic function is employed to approximate the likelihood~\cite{jaakkola1997variational}, $\rho(x) \geq \rho(\xi)e^{\frac{x-\xi}{2}-\lambda(\xi)(x^{2}-\xi^{2})}$, where $\lambda(\xi)=\frac{1}{2\xi}(\rho(\xi)-\frac{1}{2})$, and $\xi$ is an auxiliary variable needs to be adjusted to make the bound tight at $x=\pm\xi$. Moreover, by assuming $||\B{\theta}||_{2}\leq A$ and $||\B{x}_{j}||_{2} \leq B$, we have
$ -\log\big[1+\exp(2\alpha\B{\theta}\B{x}_{j}^{\top}) \B{x}_{j}\B{x}_{j}^{\top}\big] \geq  - \exp(2\alpha\B{\theta}\B{x}_{j}^{\top})\B{x}_{j}\B{x}_{j}^{\top} \geq -\exp (2\alpha AB)B^{2}$. As we assume $\B{m}$ and $\B{\Sigma}$ are bounded, it is reasonable to infer that $\B{\theta}$ is bounded. By normalizing $\B{x}_{j}$, we can make $\B{x}_{j}$ bounded. Then, we have the following lower bound of the log-likelihood function in Eq.~\eqref{eq:likelihood}:
\begin{align}
\small
\label{eq:approxilikelihood}
  \log p(\mathcal{Y}, \mathcal{S}|\B{\theta}) \geq
  \underbrace{\sum_{a_{i}\in \mathcal{S}}\big[(2y_{i}-1)\frac{\B{\theta}\B{x}_i^{\top}}{2}
    -\lambda(\xi_i)(\B{\theta}(\B{x}_i^{\top}  \B{x}_i)\B{\theta}^{\top})+2\alpha \B{\theta}\B{x}_i^{\top} + \phi(\xi_{i})\big]}_{\log h(\B{\theta}, \B{\xi})}+ \mbox{const.}
\end{align}
where $\phi(\xi_{i})=\log \rho(\xi_i) - \frac{\xi_{i}}{2}+\lambda(\xi_i)\xi_i^2$. The optimal variational distribution of $\B{\theta}$ is as follows: $\log q^{\ast}(\B{\theta}) \approx \mathbb{E}\big[\log h(\B{\theta}, \B{\xi})\big] + \mathbb{E}\big[\log p(\B{\theta})\big] + \mbox{const}.$
Due to model conjugacy, we can know that $q(\B{\theta})$ shall follow a Gaussian distribution $\mathcal{N}(\B{m}_{post}, \B{\Sigma}_{post})$, where the mean and variance are as follows:
\begin{align}
  \B{\Sigma}_{post}^{-1} &= \B{\Sigma}^{-1} +2\sum_{a_{i}\in \mathcal{S}}\lambda(\xi_i)\B{x}_i^{\top}\B{x}_i, \label{eq:sigmaupdate}\\
  \B{m}_{post} &= \B{\Sigma}_{post}\big[\B{\Sigma}^{-1}\B{m}+\sum_{a_{i}\in \mathcal{S}}(y_i+2\alpha -\frac{1}{2})\B{x}_i\big].\label{eq:meanupdate}
\end{align}
Since no prior has been assigned to $\xi_{i}$, the optimal value of $\xi_{i}$ can be derived by maximizing the expected log-likelihood function:$\ell(\xi_{i})=\mathbb{E}[\log p(\mathcal{Y}, \mathcal{S}|\B{\theta}, \xi_{i})]$. Taking the derivation of $\ell(\xi_{i})$ with respect to $\xi_{i}$ and setting it to zero, the optimal value of $\xi_{i}$ can be obtained as follows:
\begin{equation}
  \xi_i = \sqrt{\B{x}_i(\B{\Sigma}_{post}+\B{m}_{post}^{\top}\B{m}_{post})\B{x}_i^{\top}}.
  \label{eq:xiupdate}
\end{equation}

We employ Thompson sampling (TS) to update the model parameters by balancing exploration and exploitation. The details of the TS algorithm are summarized in Algorithm~\ref{alg:TS}. In standard TS method, it is required to sample from the true posterior of model parameter $\B{\theta}$. As the logistic likelihood function is not conjugate with the Gaussian prior, we propose to sample from the approximated poster distribution $q(\B{\theta})$. Once completing the sampling of $\B{\theta}$, the DPP kernel matrix $\B{L}$ is fixed, and we can select the optimal super arm $\mathcal{S}$ by maximizing $f_{\B{\theta}}(\mathcal{S})=\prod_{a_{i}\in \mathcal{S}}p_i \det(\B{L}_{[\mathcal{S}]})$. Following~\cite{chen2018fast}, we employ the fast gready MAP inference algorithm to obtain the optimal super arm. The details of the greedy algorithm are summarized in Algorithm~\ref{alg:dppgreedy}.

\begin{algorithm}[H]
   \caption{DPP Greedy Search $\mathcal{S} \leftarrow \mathcal{O}(\widehat{\B{\theta}},\B{X}_t)$}
   \label{alg:dppgreedy}
\begin{algorithmic}
   \STATE \textbf{Startup}: Construct $\B{L}, \B{p}$ according to $\widehat{\B{\theta}},\B{X}_t$
   \STATE \textbf{Initialize} $\mathbf { c } _ { i } = [ ]$ , $d _ { i } ^ { 2 } = \B { L } _ { i i }$ , $j = \arg \max _ { i \in Z } \log \left( d _ { i } ^ { 2 } \right) + \log (p_i)$ , $\mathcal{S} = \{ j \}$.
   \FOR{$k=0$ {\bfseries to} $K$}
   \FOR{$i \in Z \backslash \mathcal{S}$}
   \STATE $e _ { i } = \left( \mathbf { L } _ { j i } - \left\langle \mathbf { c } _ { j } , \mathbf { c } _ { i } \right\rangle \right) / d _ { j }$
   \STATE $\mathbf { c } _ { i } = \left[ \begin{array} { l l } { \mathbf { c } _ { i } } & { e _ { i } ] , d _ { i } ^ { 2 } = d _ { i } ^ { 2 } - e _ { i } ^ { 2 } } \end{array} \right.$
   \ENDFOR
   \STATE $j = \arg \max _ { i \in Z \backslash Y _ { \mathrm { g } } } \log \left( d _ { i } ^ { 2 } \right) + \log(p_i), \mathcal{S} = \mathcal{S} \cup \{ j \}$
   \ENDFOR
   \STATE \textbf{Return} $\mathcal{S}  $
\end{algorithmic}
\end{algorithm}

\section{Regret Analysis}
We consider a model involving a set of actions $\B{S}$ and a set of functions $\mathcal { F } = \left\{ f _ { \B{\theta} } : \B { S } \mapsto \mathbb { R } | \B{\theta} \in \B{\Theta} \right\}$ indexed by a random variable $\B{\theta}$ which belongs to an index set $\B{\Theta}$.
At each time t, a random subset $\B{S}_t \subseteq \B{S}$ is presented and an action $\mathcal{S}_t \in \B{S}_t$ is selected after which the reward $R_t$ is gained. We define the reward function as:  $\mathbb{E}[R_t] \triangleq f_{\B{\theta}}(\mathcal{S}_t)=\prod_{a_{i}\in \mathcal{S}_t}p_i \det(\B{L}_{[\mathcal{S}_t]})=\prod_{a_{i} \in \mathcal{S}_t }p_{i} r_{i}^{2\alpha}  \det ( \B{X}_{[\mathcal{S}_{t}]}\B{X}_{[\mathcal{S}_{t}]}^{T})$, and define the reward at trial $t$ as: $R_t = f_{\B{\theta}}(\mathcal{S}_t)+\epsilon_t$. Therefore, we have $\mathbb{E}[\epsilon_t] = 0$.
In addition we assume $\forall f_{\B {\theta}} \in \mathcal{F}, \forall \mathcal{S}_t \in \B{S}, f_{\B{\theta}}(\mathcal{S}_t) \in [0, C]$. For a combinatorial recommendation policy $\pi$, we can define the Bayesian risk bound as follows:
\begin{equation}
  Regret(T, \pi) = \sum_{t = 1 } ^ { T } \mathbb { E } \left[ \max _ { s\in \B { S }_t  } f _ { \B{\theta} } ( s ) - f _ { \B{\theta} } \left( \mathcal{S} _ { t } \right) \right].
\end{equation}

To perform the regret analysis, we first introduce the following two Lemmas.  
The detailed proofs can be found in the supplementary material.
\begin{lemma}
For all $T \in \mathbb { N }$, $\alpha_0 > 0$ and $\delta \leq 1 / 2 T$,
\begin{align}
  Regret \left( T , \pi ^ { \mathrm { TS } } \right) \leq& 4 \sqrt { \operatorname { dim } _ { M } \left( \mathcal { F } , T ^ { - 1 } \right) \beta _ { T } ^ { * } ( \mathcal { F } , \alpha_0 , \delta ) T }
  +1 + \left[ \operatorname { dim } _ { M } \left( \mathcal { F } , T ^ { - 1 } \right) + 1 \right] C, \nonumber
\end{align}
where $ \operatorname { dim }_{ M } \left( \mathcal { F } , T ^ { - 1 } \right) $ is the $\epsilon$-dimension, $\beta _ { T } ^ { * } ( \mathcal { F } ,  \alpha_0, \delta  ) : = 8 \ln \left( N \left( \mathcal { F } , \alpha_0 , \| \cdot \| _ { \infty } \right) / \delta \right) + 2 \alpha t \left(\frac{15}{2} C +  \ln ( 2 t ^ { 2 } / \delta ) \right)$, and $N \left( \mathcal { F } , \alpha_0 , \| \cdot \| _ { \infty } \right)$ denotes the $\alpha_0$-covering number of $\mathcal{F}$.\nonumber
\end{lemma}
\begin{lemma}
  Suppose $\B{\Theta} \subset \mathbb{R}^{d}$, and $\left|f_{\B{\theta}}(\mathcal{S}_{t}) - f_{\B{\theta}^{\star}} (\mathcal{S}_{t})\right| \leq \left| h(\B{\theta} - \B{\theta}^{\star})^\top  \phi (\mathcal{S} _t) \right| $, where $\phi(\mathcal{S}_{ t }) = \sum _ { i \in \mathcal{S} _ { t } } \B{x}_{i}$ and $h$ is a constant. Assume there exist constants $ \gamma$, $S_0$ such that $\forall \mathcal{S}_t \in \B{S}$ and $\B{\theta} \in \B{\Theta}$, $\| \B{\theta} \| _ { 2 } \leq S_0$, and $\| \phi ( \mathcal{S}_t ) \| _ { 2 } \leq \gamma $. Then we have
\begin{eqnarray}
\small
  \operatorname { dim } _ { M } ( \mathcal { F } , \epsilon )  \leq 3 d  \frac { e } { e - 1 } \ln \left\{ 3  \left( \frac { 2 S_0 h \gamma } { \epsilon } \right) ^ { 2 } \right\} + 1.
\end{eqnarray}
\end{lemma}
According to Lemma 1, in our problem, $C=1$, and we can choose $\alpha_0 = 1/T^2, \delta = 1/T^2$. Then, the Bayesian risk bound of the proposed method is given by the following Theorem.
\begin{theorem}
When T is sufficient large, the Bayesian risk bound of DC$^2$B is
\begin{eqnarray}
  Regret \left( T , \pi ^ { T S } \right) =  O \left( d \ln \left( \alpha s e^{2 \alpha s} T \sqrt{d} \right) \sqrt { T } \right),
\end{eqnarray}
where $s$ is the number of items in a selected set, $d$ is the dimension of $\B{\theta}$.
\end{theorem}
\begin{proof}
We first assume $\| \B{\theta} \|_2 \leq 1$, $\| \B{x}\|_2 \leq 1$, and introduce the following inequalities: (1) Mean Value Theorem: we have $|\B{\theta} \B{x}| \leq \|\B{\theta}\|_2 \|\B{x}\|_2 \leq 1$, then
$
  |p-p^{\star}| = | \rho(\B{\theta}^\top \B{x})-\rho(\B{\theta}^{\star \top} \B{x}) |
  = | \rho'(\xi)(\B{\theta} - \B{\theta}^{\star})^\top \B{x}|
  \leq \frac{1}{4} \|\B{\theta} - \B{\theta}^{\star}\|_2 \|\B{x}\|_2 $, and
$
  |r^2-r^{\star 2}| = |\exp(2\alpha \B{\theta}^\top \B{x} )-\exp(2\alpha \B{\theta}^{\star \top} \B{x})|
  =  |\exp(2\alpha \zeta)2(\B{\theta} - \B{\theta}^{\star})^\top \B{x}|
  \leq  2\alpha e^{2\alpha} \|\B{\theta} - \B{\theta}^{\star} \|_2 \|\B{x}\|_2$, where $\rho'(x) = \rho(x)(1-\rho(x)) \leq \frac{1}{4}$, $0 \leq \xi \leq 1$, $0 \leq \zeta \leq 1$;
(2) Gram Inequality: $
  |\operatorname { det }  ( \B{X}_{[\mathcal{S}]}^{\top} \B{X}_{[\mathcal{S}]} ) | = \left|\det \left( G \left( \B{x} _ { 1 } , \cdots , \B{x} _ { s } \right)\right) \right|
  \leqslant   \left\| \B{x} _ { 1 } \right\|_2 ^ { 2 } \cdots \left\| \B{x} _ { s } \right\|_2 ^ { 2 } \leq 1,
$ where $[G \left( \B{x} _ { 1 } , \cdots , \B{x} _ { n } \right)]_{i, j} = \B{x}_i^\top \B{x}_j$ defines a gram matrix;
(3) Triangle inequality: $
  |x_1x_2 - y_1y_2| = |x_1x_2 - y_1x_2 + y_1x_2 - y_1y_2| \leq |x_1-y_1||x_2| + |y_1||x_2-y_2|.
$
Based on these inequalities, we have,
\begin{align}
\small
  |f_{\B{\theta}}(\mathcal{S}_t) - f_{\B{\theta}^{\star}}(\mathcal{S}_t)|
  &= | \operatorname { det } ( \B{X} _ { [\mathcal{S}_t] } ^ { \top } \B{X} _ { [\mathcal{S}_t] } )| | \prod_{i=1}^{s} p_i r_i ^{2\alpha} - \prod_{i=1}^{s} p_i ^{\star} r_i ^{\star 2 \alpha} |  \nonumber\\
  &\leq (\frac{8\alpha + 1}{4} ) e^{2 \alpha s}|(\B{\theta} - \B{\theta} ^{\star})^\top \sum_{i=1}^{s} \B{x}_i|
  \leq \frac{8\alpha +1}{4} s e^{2 \alpha s} \sqrt{d} \| \B{\theta} - \B{\theta} ^{\star} \|_{\infty},
\end{align}
where we use inequality $\|\B{\theta}\|_2 \leq \sqrt{d} \|\B{\theta}\|_{\infty}$. According to Eq.(13), an $\alpha_0$-covering of $\mathcal{F}$ can therefore be attained through an $(\alpha_0 / \gamma)$-covering of $\B{\Theta} \subset \mathbb { R } ^ { d }$, where $\gamma = \frac{8\alpha +1 }{4} s e^{2 \alpha s} \sqrt{d} $. Evenly divide $\mathbb { R } ^ { d }$ in each dimension, we can obtain $N \left( \mathbb { R }^d , \alpha_0 , \| \cdot \| _ { \infty } \right) = ( 1 / \alpha_0 ) ^ { d }$. Then, we have
\begin{align}
\small
  N \left( \mathcal { F } , \alpha_0 , \| \cdot \| _ { \infty } \right) = ( \gamma / \alpha_0 ) ^ { d }
  = \left( \frac{(8\alpha + 1)se^{2 \alpha s} \sqrt{d}}{4\alpha_0} \right)^d.
\end{align}
In our problem, $S_0=1$, $h= \frac{8\alpha +1 }{4} e^{2\alpha s}$ and $\gamma = s$. According to Lemma 2 and Eq.(13), we have the following bound on $\operatorname { dim } _ { M } \left( \mathcal { F } , T ^ { - 1 } \right)$:
\begin{eqnarray}
\small
  \operatorname { dim } _ { M } ( \mathcal { F } , T^{-1}) \leq 3 d \frac { e } { e - 1 } \ln \left\{  3 \left( \frac { (8\alpha +1 )  e^{2 \alpha s} T s } { 2 } \right) ^ { 2 } \right\} + 1.
\end{eqnarray}
Let $\alpha_0 = 1/T^2$, $\delta = 1/T^2$, $C=1$. When $T$ is sufficient large, the second part of  $\beta _ { T } ^ { * } ( \mathcal { F } ,  \alpha_0, \delta  )$  will decrease to zero. After some calculation together with above two bounds, we can finish the proof.
\end{proof}

The upper bound in Theorem 1 mainly depends on the dimensionality of model parameter $d$, the size of recommended item set $s$, and the quality controlling parameter $\alpha$. Here, $d$ describes the model complexity. As $d$ increases, $\B{\theta}$ is able to model more complex scenarios. However, a sophisticated model would cause over-fitting, resulting in poor performances. Therefore, the regret bound would be high, when $d$ is large. The Proposition 9 in~\cite{russo2014learning} gives the Bayesian risk bound for non-combinatorial bandit methods as $O ( r d \sqrt{T} \log  (r T ) )$, where $r$ is a parameter determined by the reward function. By simply repeating the recommendation $s$ times to get a set of items, the bound would be $O (s r d \sqrt{T} \log (r T ) )$. In DC$^2$B, if we set $\alpha = 1$, the Bayesian regret bound would be $O(d \sqrt{T} \log( s e^{2s} \sqrt{d} T) )$, which is slightly different from multiplying $s$ to the bound of non-combinatorial methods. This is because our reward function also takes the recommendation diversity into account. As $\alpha$ controls the impacts of item qualities, the increase of $\alpha$ would increase the risks caused by the estimation of item qualities. Thus, the regret will grow as $\alpha$ increases. 

\section{Experiments}

\subsection{Experimental Settings}
\noindent\textbf{Datasets}: The experiments are performed on the following datasets: Movielens-100K, Movielens-1M\footnote{https://grouplens.org/datasets/movielens/}, and Anime\footnote{https://www.kaggle.com/CooperUnion/anime-recommendations-database}. Movilens-100K contains 100,000 ratings given by 943 users to 1,682 movies, and Movielens-1M contains 1,000,209 ratings given by 6,040 users to 3,706 movies. There are 18 movie categories in both Movielens datasets. We denote these two datasets by ML-100K and ML-1M, respectively. For Anime dataset, there are 7,813,737 ratings given by 73,515 users to 11,200 animes, and there are 44 anime categories. Following~\cite{rendle2009bpr}, we keep the ratings larger than 3 as positive feedback on ML-100K and ML-1M datasets, and keep the ratings larger than 6 as positive feedback on the Anime dataset. Table~\ref{tab:dataset} summarizes the statistics of the experimental datasets, where movies and animes are ``items". In these datasets, each item may belong to multiple categories.
\begin{table}
    \centering
    \small
    \caption{Statistics of the experimental datasets.}
    \label{tab:dataset}
    \begin{tabular}{|l|c|c|c|c|c|} \hline
    Datasets & \# Users & \# Items  & \# Interactions & \# Categories & Density \\\hline
    ML-100K & 942 & 1,447 & 55,375 & 18 & 4.06\%  \\ \hline
    ML-1M & 6,038 & 3,533 & 575,281 & 18 & 2.70\%\\ \hline
    Anime & 69,400 & 8,825 & 5,231,117 & 44 & 0.85\%\\ \hline
    \end{tabular}
    \vspace{-10px}
\end{table}

\noindent\textbf{Setup and Metrics}: We employ the unbiased offline evaluation strategy~\cite{li2011unbiased} to evaluate the recommendation methods. Following~\citep{zhao2013interactive,qin2014contextual,wang2017interactive}, we assume that users' ratings on items recorded in our experimental datasets are not biased by the recommender system, and these records can be regarded as unbiased user feedback in our experimental settings. In the experiments, we randomly partition each dataset into two non-overlapping sets, by randomly sampling 80\% of the users for training and using the remaining 20\% users for testing. Moreover, we employ BPRMF~\cite{rendle2009bpr} to learn the embeddings of items, which are used as the contextual features of arms. Empirically, we set the dimensionality of the item embeddings to 10. 
As users are usually interested in a few top-ranked recommendation items, we adopt Precision@$N$ to evaluate the recommendation accuracy~\cite{shi2014collaborative}. Specifically, $N$ is set to 10, 30, and 50. We also evaluate the average recommendation diversity of each method over all recommendation trials, by using the intra-list distance (ILD)~\cite{zhang2008avoiding} metric as follows:
\begin{equation}
\small
  \mbox{Diversity}=\frac{1}{T}\sum_{t=1}^{T}\bigg[\frac{2}{|\mathcal{S}_t|(|\mathcal{S}_t|-1)}\sum_{a_{i}\in \mathcal{S}_t}\sum_{a_{j}\in \mathcal{S}_t, i\neq j}(1 - sim_{ij} )\bigg],
\end{equation}
where $\mathcal{S}_{t}$ is recommended item set at trial $t$, $|\mathcal{S}_t|$ denotes the size of $\mathcal{S}_{t}$, $T$ is the total number of recommendation trials, $sim_{ij}$ denotes the similarity between $a_{i}$ and $a_{j}$. As an item may belong to multiple item categories, we define the item similarity $sim_{ij}$ by using the Jaccard similarity of the categories of two items. For these accuracy and diversity metrics, we first compute the value for each user, and then report the averaged value over all users. Following~\cite{cheng2017learning}, we also employ F-measure to evaluate the performances of different methods on trading-off between accuracy and diversity, where \emph{F-measure=2*accuracy*diversity / (accuracy+diversity)}.

\textbf{Evaluated Recommendation Methods}: As the training users are non-overlapping with the testing users, the recommendation algorithms~\cite{shi2014collaborative} designed for warm-start settings are not suitable as baselines. In this paper, we compare DC$^2$B with the following recommendation methods: 
(1) \textbf{LogRank}: In this method, we  define the quality score of each arm $a_{i}$ as $r_{i}= 1 / (1 +\exp(-\bar{\B{u}}\B{x}_{i}^{\top}))$, where $\bar{\B{u}}$ is the mean of the user embeddings learnt from the training data. Then, the $|\mathcal{S}_{t}|$ available arms with the highest quality scores are selected as a super arm $\mathcal{S}_{t}$ for recommendation at trial $t$; (2) \textbf{MMR}: This method employs MMR strategy~\cite{carbonell1998use} to promote the recommendation diversity. At trial $t$, this method sequentially selects an available arm with the largest maximal marginal relevance score into $\mathcal{S}_{t}$. The maximal marginal relevance score is defined as $\tilde{r}_{i}=\alpha r_{i} - \frac{(1-\alpha)}{|\mathcal{S}_{t}|}\sum_{j \in \mathcal{S}_{t}}sim(\B{x}_{i}, \B{x}_{j})$, where $r_{i}$ is the arm quality defined in the LogRank method, and $sim(\B{x}_{i}, \B{x}_{j})$ is the Cosine similarity between $\B{x}_{i}$ and $\B{x}_{j}$; (3) \textbf{$\epsilon$-Greedy}: This method randomly adds an available arm into $\mathcal{S}_{t}$ with probability $\epsilon$, and adds the arm with highest quality into $\mathcal{S}_{t}$ with probability $1-\epsilon$. The item quality is defined the same as in LogRank method; (4) \textbf{DPP$^{map}$}~\citep{chen2018fast}: This non-interactive method uses determinantal point process to promote recommendation diversity. The item quality is defined the same as in LogRank. (5) \textbf{C$^2$UCB}~\cite{qin2014contextual}: This methods integrates the LinUCB framework with an entropy regularizer to promote diversity for interactive recommendation; 
(6) \textbf{EC-Bandit}~\cite{qi2018bandit}: This bandit method is based on Thompson sampling framework and developed for interactive recommendation with users' implicit feedback. In this method, the user need to interacts with the recommender $|\mathcal{S}_{t}|$ times to generate the recommended item set at trial $t$. For all methods, we empirically set the size of $\mathcal{S}_{t}$ to 10 in each trial. A validation set is sampled from training data to choose hyper-parameters. The best parameter settings for each method are as follows. $\alpha$ is set to 0.9 for MMR. $\epsilon$ is set to 0.1 for $\epsilon$-Greedy, and $\theta$ is set to 0.6 for DPP$^{map}$. In C$^2$UCB, we set $\lambda_0=100$, $\lambda = 0.1$, and $\sigma = 1$. In EC-Bandit, we set the parameter $\lambda = 1$. For DC$^2$B, we empirically set $\alpha =3$, and $\lambda = 1$, on all datasets.

\begin{table}
    \centering
    \small
    \caption{Recommendation performances of different algorithms. The best results are in \textbf{bold faces} and the second best results are \underline{underlined}. }
    \label{tab:precision}
    \begin{tabular}{|l|l|c|c|c|c|c|c|c|} \hline
    Datasets &Metrics &LogRank & MMR & $\epsilon$-Greedy & DPP$^{map}$ & C$^2$UCB & EC-Bandit & DC$^2$B \\\hline
    \multirow{5}{*}{ML-100K}
    & Prec.@10 & 0.3548 & \textbf{0.3665} & 0.3421 & \textbf{0.3665} & 0.3633 & 0.2128 & \underline{0.3649} \\\cline{2-9}
    & Prec.@30 & 0.2872 & 0.2872 & 0.2792 & 0.2846 & \textbf{0.3415} & 0.1633 & \underline{0.3211} \\ \cline{2-9}
    & Prec.@50 & 0.2507 & 0.2499 & 0.2433 & 0.2554 & \textbf{0.3146} & 0.1453 & \underline{0.2882}\\ \cline{2-9}
    & Diversity & 0.8024 & \underline{0.8151} & 0.8145 & 0.7985 & 0.7827 & \textbf{0.8356} & 0.8118 \\ \cline{2-9}
    & F-meas. & 0.3820 & 0.3825 & 0.3747 & 0.3870 & \textbf{0.4488} & 0.2476 & \underline{0.4254} \\
    \hline\hline
    \multirow{5}{*}{ML-1M}
    & Prec.@10 & \textbf{0.3785} & 0.3754 & 0.3631 & \underline{0.3764} & 0.3418 & 0.2160 & \textbf{0.3785} \\ \cline{2-9}
    & Prec.@30 & \underline{0.3204} & 0.3173 & 0.3084 & 0.3173 & 0.3192 & 0.1750 & \textbf{0.3401} \\\cline{2-9}
    & Prec.@50 & 0.2841 & 0.2824 & 0.2745 & 0.2807 & \underline{0.2998} & 0.1611 & \textbf{0.3117} \\\cline{2-9}
    & Diversity & \underline{0.8516} & \textbf{0.8531} & 0.8462 & 0.8174 & 0.8319 & 0.8326 & 0.8367 \\ \cline{2-9}
    & F-meas. & 0.4261 & 0.4221 & 0.4145 & 0.4179 & \underline{0.4408} & 0.2700 & \textbf{0.4542}\\\hline
   \hline
    \multirow{5}{*}{Anime}
    & Prec.@10 & \underline{0.3141} & \textbf{0.3157} & 0.2867 & \textbf{0.3157}& 0.0095 & 0.1733 & 0.3003 \\\cline{2-9}
    & Prec.@30 & 0.2527 & 0.2534 & 0.2366 & \underline{0.2541} & 0.1116 & 0.1326 & \textbf{0.2666} \\\cline{2-9}
    & Prec.@50 & 0.2165 & \underline{0.2178} & 0.2025 & 0.2164 & 0.1518 & 0.1168 & \textbf{0.2419} \\\cline{2-9}
    & Diversity & 0.8323 & \underline{0.8495} & \textbf{0.8521} & 0.8414 & 0.5031 & 0.8460 & 0.8355 \\\cline{2-9}
    & F-meas. & 0.3436 & \underline{0.3467} & 0.3272 & 0.3443 & 0.2332 & 0.2053 & \textbf{0.3752}\\
    \hline
    \end{tabular}
    \vspace{-10px}
\end{table}

\subsection{Results and Discussion}
The recommendation accuracies and diversity of different algorithms are summarized in Table~\ref{tab:precision}. As shown in Table~\ref{tab:precision}, the proposed DC$^2$B method usually achieves the best recommendation accuracy (i.e., Precision@$N$) on ML-1M and Anime datsets, and achieves the second best accuracy on ML-100K dataset.
For example, on Anime dataset, DC$^2$B significantly outperforms C$^2$UCB and EC-Bandit by 59.35\% and 107.11\%, and achieves 11.73\%, 11.07\%, 19.46\%, and 11.78\% better performances than LogRank, MMR, $\epsilon$-Greedy, and DPP$^{map}$, in terms of Precision@50. These results indicate that DC$^2$B is more effective than baseline methods on large and sparse dataset. 
Moreover, we also note the combinatorial bandit methods C$^2$UCB and DC$^2$B significantly outperform EC-Bandit. One potential reason is that the combinatorial methods employ the user's feedback on a set of items to update model parameters. However, EC-Bandit uses the user's feedback on a single item to update model parameters. The parameter learning of C$^2$UCB and DC$^2$B is more stable than that of EC-Bandit, thus C$^2$UCB and DC$^2$B can achieve better recommendation accuracy. In addition, we can note that the non-interactive methods MMR, and $\epsilon$-Greedy usually achieves slightly higher recommendation diversity than DC$^2$B, and DC$^2$B attains better recommendation diversity than DPP$^{map}$ on ML-100K and ML-1M datasets. Comparing with interactive methods, Table~\ref{tab:precision} indicates that the recommendation diversity of DC$^2$B is higher than that of C$^2$UCB on all datasets, and EC-Bandit achieves higher recommendation diversity than DC$^2$B on ML-100K and Anime datasets.

For better understanding the results, F-measure is used to study the effectiveness of each recommendation algorithm in balancing the recommendation accuracy and diversity. Here, we use Precision@50 and Diversity to compute the F-measure. As shown in Table~\ref{tab:precision}, we can note that DC$^2$B achieves the best F-measure value on ML-1M and Anime datasets, and the second best F-measure value on ML-100K dataset. In addition, we summarize the relative improvements of DC$^2$B over baseline methods on Precision@50, Diversity, and F-measure in Table~\ref{tab:improvements}. These results demonstrate that the proposed DC$^2$B method is more effective in balancing the recommendation accuracy and diversity than the baseline methods, especially on larger and sparser datasets.

Moreover, we also evaluate the impacts of $\alpha$ and the size of super arm $|\mathcal{S}_{t}|$ on the performances of DC$^2$B, on ML-100K dataset. The parameter $\alpha$ is varied in $\{0.01, 0.1, 1, 3, 5, 10, 100\}$. Figure~\ref{fig:alphaTrend} shows the performances of DC$^2$B with respect to different settings of $\alpha$. As shown in Figure~\ref{fig:alphaTrend}, the recommendation accuracy in terms of Precision@50 firstly increases with the increase of $\alpha$. When $\alpha$ is larger than 5, the recommendation accuracy of DC$^2$B drops drastically by changing $\alpha$ to 10 and 100. We can also note that the diversity value decreases with the increase of $\alpha$, because larger $\alpha$ makes DC$^2$B focus more on the item qualities in generating recommendations. Overall, the results in Figure~\ref{fig:alphaTrend} indicate that $\alpha$ can effectively control the item qualities and item diversities when generating recommendations. Additionally, we vary the size of recommendation list $|\mathcal{S}_{t}|$ at each trial in $\{5, 10\}$. Figure~\ref{fig:sizeTrend} summarizes the accuracy of DC$^2$B with respect to different sizes of super arm. As shown in Figure~\ref{fig:sizeTrend}, larger super arm size tends to results better recommendation accuracy, when enough number of interactions (e.g., more than 3 interactions) between the user and the recommender system have been performed. This is because the model updating based on the user's feedback on a larger set of items is expected to be more stable and accurate. Moreover, the average recommendation diversity with respect to $|\mathcal{S}_{t}|=5$ and $|\mathcal{S}_{t}|=10$ are 0.8110 and 0.8118, respectively. This indicates that $|\mathcal{S}_{t}|$ does not have significant impacts on the recommendation diversity.

\begin{table}
    \centering
    \scriptsize
    \caption{Relative improvements of DC$^2$B over baselines. The positive improvements are highlighted in \textbf{bold}.}
    \label{tab:improvements}
    \begin{tabular}{|l|c|c|c|c|c|c|c|c|c|} \hline
    \multirow{2}{*}{Methods} & \multicolumn{3}{c|}{ML-100K} & \multicolumn{3}{c|}{ML-1M} & \multicolumn{3}{c|}{Anime} \\\cline{2-10}
    & Prec@50 & Div. & F-m. & Prec@50 & Div. & F-m. & Prec@50 & Div. & F-m.\\\hline
    LogRank & \textbf{+14.96\%} & \textbf{+1.17\%} & \textbf{+11.36\%} & \textbf{+9.71\%} & -1.75\% & \textbf{+6.59\%} & \textbf{+11.73\%} & \textbf{+0.38\%} & \textbf{+9.20\%}\\ \hline
    MMR & \textbf{+15.33\%} & -0.40\% & \textbf{+11.22\%} & \textbf{+10.38\%} & -1.92\% & \textbf{+7.60\%} & \textbf{+11.07\%} & -1.65\% & \textbf{+ 8.22\%}\\ \hline
    $\epsilon$-Greedy & \textbf{+18.45\%} & -0.33\% & \textbf{+13.53\%} & \textbf{+13.55\%} & -1.12\% & \textbf{+9.58\%} & \textbf{+19.46\%} & -1.95\% & \textbf{+14.67\%}\\\hline
    DPP$^{map}$ & \textbf{+12.84\%} & \textbf{+1.67\%} & \textbf{+9.92\%} & \textbf{+11.04\%} & \textbf{+2.36\%} & \textbf{+8.69\%} & \textbf{+11.78\%} & -0.70\% &  \textbf{+8.97\%}\\\hline
    C$^2$UCB & -8.39\% & \textbf{+3.72\%} & -5.21\% & \textbf{+3.97\%} & \textbf{+0.58\%} & \textbf{+3.04\%} & \textbf{+59.35\%} & \textbf{+66.07\%} & \textbf{+60.89\%}\\\hline
    EC-Bandit & \textbf{+98.35\%} & -2.85\% & \textbf{+71.81\%} & \textbf{+93.48\%} & \textbf{+0.49\%} & \textbf{+68.22\%} & \textbf{+107.11\%} & -1.21\% & \textbf{+82.76\%}\\\hline
    \end{tabular}
\end{table}

\begin{figure*}[!ht]
    \centering
   \subfigure[]{
        \includegraphics[height=1.5in]{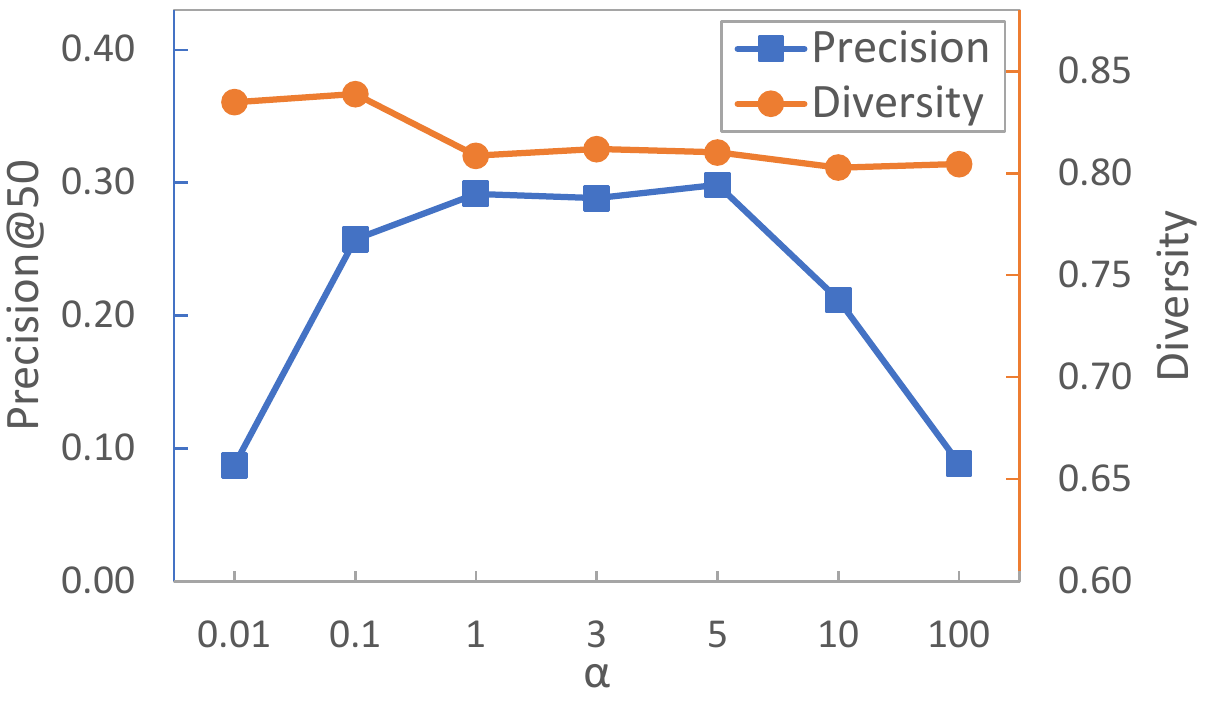}
        \label{fig:alphaTrend}
    }
    \subfigure[]{
        \includegraphics[height=1.5in]{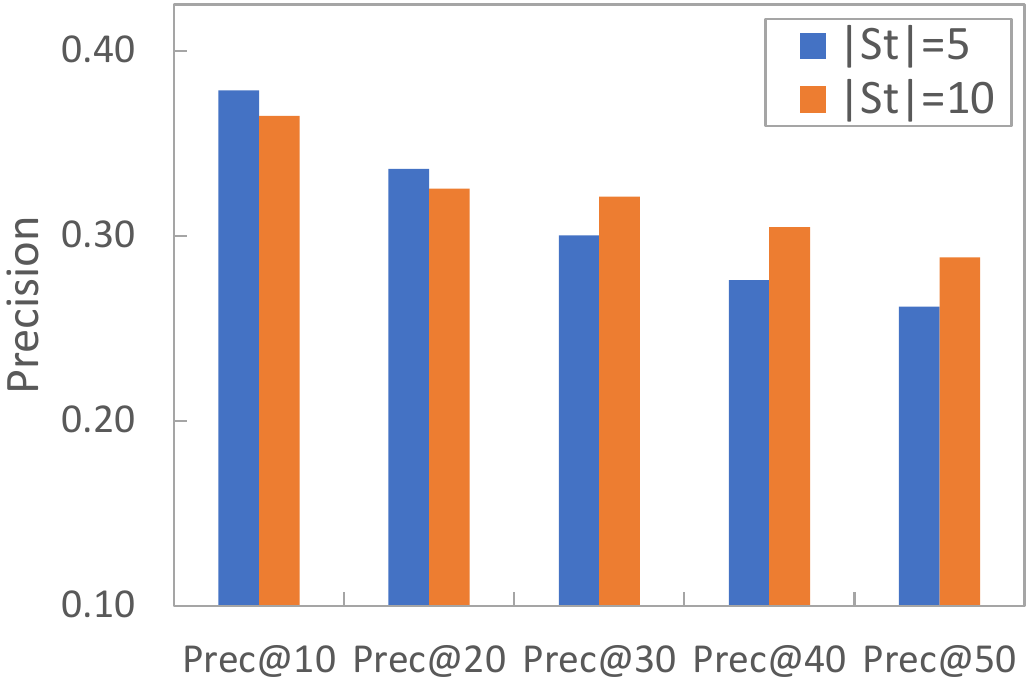}
        \label{fig:sizeTrend}
    }
    \vspace{-13pt}
   \caption{Performance trend of DC$^2$B with respect to different settings of $\alpha$ and $|\mathcal{S}_{t}|$.}
    \vspace{-18pt}
    \label{fig:diversity}
\end{figure*}

\section{Conclusion}
This work proposes a novel bandit learning method, namely Diversified Contextual Combinatorial Bandit (DC$^2$B), for interactive recommendation based on users' implicit feedback. Specifically, DC$^2$B is a full Bayesian recommendation framework, which enables the interactions between recommender system and the user, and employs determinantal point process (DPP) to promote the recommendation diversity. We have proposed Thompson sampling-type optimization algorithm to iteratively learn the model parameters, and conducted regret analysis to provide theoretical guarantee for DC$^2$B. Moreover, empirical experiments on real datasets also demonstrate the effectiveness of the proposed DC$^2$B in balancing the recommendation accuracy and diversity. The future work will focus on the following potential directions. First, we would like to develop more complex DPP kernels and more efficient DPP inference algorithms for interactive recommender systems. Second, we are also interested in developing more sophisticated models to describe the user engagements on the recommendation results. Last but not least, we will apply DC$^2$B to online recommendation scenarios to evaluate its online performances.

\bibliographystyle{unsrt}
\bibliography{references}

\appendix
\appendixpage
\addappheadtotoc

\section{Proof of Lemma 1 and Lemma 2}

\subsection{Preliminaries}
For simplicity, we first use $f_{\theta}$ to represent $f_{\theta}(S_t)$, and define $L _ { t } ( f ) = \sum _ { k=1 } ^ { t - 1 } \left( f(S_k)  - R _ { k } \right) ^ { 2 }$, $\hat{f}_{t}^{LS}  =  \arg \min_{f \in \mathcal{F}} L_{t} (f)$, and $\| f \| _ { E _ { t } } ^ { 2 } = \sum _ { k = 1 } ^ { t - 1 } f ^ { 2 } \left( S _ { k } \right)$. Then, we introduce the following two important inequalities.

\textbf{Martingale Exponential Inequalities.} \noindent \emph{Consider random variables $\left( Z _ { n } | n \in \mathbb { N } \right)$ adapted to the filtration $\left( \mathcal{ H } _ { n } : n = 0,1 , \dots \right)$, Assume that $\mathbb { E } \left[ \exp \left\{ \lambda Z _ { i } \right\} \right]$ is finite for all $\lambda $. We define $\mu _ { i } = \mathbb { E } \left[ Z _ { i } | \mathcal { H } _ { i - 1 } \right]$,  $\psi _ { i } ( \lambda ) = \log \mathbb { E } \left[ \exp \left( \lambda \left[ Z _ { i } - \mu _ { i } \right] \right) | \mathcal { H } _ { i - 1 } \right]$. Then for all $x \geq 0$ , $\lambda \geq 0$ and $\forall n \in \mathbb { N } $, we have
\begin{eqnarray}
  \mathbb { P } \left( \sum _ { 1 } ^ { n } \lambda Z _ { i } \leq x + \sum _ { 1 } ^ { n } \left[ \lambda \mu _ { i } + \psi _ { i } ( \lambda ) \right] \right) \geq 1 - e ^ { - x }.
\end{eqnarray}}

\textbf{Hoeffding's Lemma.} \noindent \emph{Let $X$ be any real-valued random variable such that $E(X)=0$ and $a\leq X \leq b$ almost surely. Then, for all $\lambda \in \mathcal{R}$,
\begin{eqnarray}
  \mathbb { E } \left[ e ^ { \lambda X } \right] \leq \exp\bigg(\lambda ^ { 2 } ( b - a ) ^ { 2 } / 8\bigg).
\end{eqnarray}}

The proof of the martingale exponential inequality can be found in the Lemma 6 in Appendix A of~\cite{russo2014learning}, and the proof of the Hoeffding's Lemma can be found in Chapter 2 of in~\cite{massart2007concentration}. To prove Lemma 1, we need the following Lemma 3 and Lemma 4.

\textbf{Lemma 3.} \noindent \emph{For any $\delta > 0$ and $f : \mathcal { A } \mapsto \mathbb { R }$, with probability at least $1-\delta$,
$
  L _ { t } ( f ) \geq L _ { t } \left( f _ { \theta } \right) + \frac { 1 } { 2 } \left\| f - f _ { \theta } \right\| _ {  E _ { t } } ^ { 2 } - 4  \log ( 1 / \delta )
$, for all $t \in \mathbb{N}.$}
\begin{proof}
We first define $\epsilon _ { t } = R _ { t } - f _ { \theta } \left( S _ { t } \right)$, $Z_t = \left( f _ { \theta } \left( S _ { t } \right) - R _ { t } \right) ^ { 2 } - \left( f \left( S _ { t } \right) - R _ { t } \right) ^ { 2 }$, where $f$ is an arbitrary function as assumed in the Lemma 1. According to the definition, we have $\mathbb { E } \left[ \epsilon _ { t } | \mathcal { H } _ { t - 1 } \right]=0$, $\sum _ { 1 } ^ { T ^ { \prime } } Z _ { t } = L _ { T + 1 } \left( f _ { \theta } \right) - L _ { T + 1 } ( f )$. With Hoeffding's Lemma and some calculation, we have
\begin{align}
  \mathbb { E } \left[ \exp \left\{ \lambda \epsilon _ { t } \right\} | \mathcal { H } _ { t - 1 } \right] &\leq \exp \left\{ \frac { \lambda ^ { 2 } \sigma ^ { 2 } } { 2 } \right\}, \\
   Z_t = - \left( f  - f _ { \theta }  \right) ^ { 2 } &+ 2 \left( f  - f _ { \theta }  \right) \epsilon _ { t }.
\end{align}
Here, we set $\sigma=1$, and it is reasonable to assume that $-1\leq \epsilon_t \leq 1$, thus $(b-a)^2 / 8 = 1/2$. Therefore, we can obtain
\begin{align}
  \mu _ { t } &= \mathbb { E } \left[ Z _ { t } | \mathcal { H } _ { t - 1 } \right] = - \left( f  - f _ { \theta } \right) ^ { 2 }, \\
\nonumber
  \psi _ { t } ( \lambda ) &= \log \mathbb { E } \left[ \exp \left( \lambda \left[ Z _ { t } - \mu _ { t } \right] \right) | \mathcal { H } _ { t - 1 } \right] \nonumber\\
  &= \log \mathbb { E } \left[ \exp \left( 2 \lambda \left( f  - f _ { \theta } \right) \epsilon _ { t } \right) | \mathcal { H } _ { t - 1 } \right]
  \leq  \frac { \left( 2 \lambda \left[ f  - f _ { \theta }  \right] \right) ^ { 2 }  } { 2 }.
\end{align}
According to the Martingale Exponential Inequality, we know that for all $x\geq 0$, $\lambda \geq 0$ and $\forall t \in \mathbb{N}$,
\begin{align}
  \mathbb { P } \bigg( \sum _ { k = 1 } ^ { t } \lambda Z _ { k } \leq x - \lambda \sum _ { k = 1 } ^ { t } \left( f  - f _ { \theta } \right) ^ { 2 }
    + \frac { \lambda ^ { 2 } } { 2 } \left( 2 f - 2 f _ { \theta } \right) ^ { 2 }   \bigg)
  \geq 1 - e ^ { - x }.
\label{eq:lemma3prob1}
\end{align}
By modifying Eq.~\eqref{eq:lemma3prob1}, we can have
\begin{align}
  \mathbb { P } \left( \sum _ { k = 1 } ^ { t } Z _ { k } \leq  \frac { x } { \lambda } + \sum _ { k = 1 } ^ { t } \left(f  - f _ { \theta }\right) ^ { 2 } \left( 2 \lambda  - 1 \right)  \right)
  \geq 1 - e ^ { - x }.
\end{align}

By setting $\lambda = \frac { 1 } { 4 }$ and $x = \log \frac { 1 } { \delta }$, we can have
\begin{align}
  \mathbb { P } \left( L _ {  t } ( f ) \geq L _ {  t } \left( f _ { \theta } \right) + \frac { \left\| f - f _ { \theta } \right\| _ {  E _ { t } } ^ { 2 } } { 2 }  - 4  \log (\frac{1}{\delta} )  \right)
  \geq 1 - e ^ { - x } = 1- \delta.
\end{align}
\end{proof}

Before showing the Lemma 4, we introduce the following results of Discretization Error.

\textbf{Discretization Error}. \noindent \emph{If $f^{\alpha}$ satisfies $\left\| f - f ^ { \alpha } \right\| _ { \infty } \leq \alpha$, then with probability at least $1-\delta$,
\begin{align}
  \ \left| \frac { \left\| f ^ { \alpha } - f _ { \theta } \right\| _ {  E _ { t } } ^ { 2 } } { 2 }  - \frac { \left\| f - f _ { \theta } \right\| _ {  E _ { t } } ^ { 2 } } { 2 }  + L _ {  t } ( f ) - L _ {  t } \left( f ^ { \alpha } \right) \right|
  \leq \alpha t \left[ \frac{15}{2} C +  \ln ( 2 t ^ { 2 } / \delta ) \right] \quad \forall t \in \mathbb { N }.
\end{align}
}
\begin{proof}
Note that $\forall f , f ^ { \alpha } \in \mathcal { F }$, $\left\| f - f ^ { \alpha } \right\| _ { \infty } \leq C$. Thus, we only consider the case $\alpha \leq C$. Firstly, we can have the following inequality,
\begin{eqnarray}
  \left| \left( f ^ { \alpha } \right) ^ { 2 } ( a ) - ( f ) ^ { 2 } ( a ) \right|
  \leq \max _ { y\in [-\alpha, \alpha] } \left| ( f ( a ) + y ) ^ { 2 } - f ( a ) ^ { 2 } \right|
  = 2 f ( a ) \alpha + \alpha ^ { 2 } \leq 2 C \alpha + \alpha ^ { 2 }.
\end{eqnarray}
Then, we can know
\begin{align}
  \left| \left( f ^ { \alpha } ( a ) - f _ { \theta } ( a ) \right) ^ { 2 } - \left( f ( a ) - f _ { \theta } ( a ) \right) ^ { 2 }\right|
  &= \left| \left[ \left( f ^ { \alpha } \right) ( a ) ^ { 2 } - f ( a ) ^ { 2 } \right] + 2 f _ { \theta } ( a ) \left( f ( a ) - f ^ { \alpha } ( a ) \right) \right| \nonumber\\
  &\leq 4 C \alpha + \alpha ^ { 2 },
  \label{eq:diseq1}
\end{align}
\begin{align}
  \left| \left( R _ { t } - f ( a ) \right) ^ { 2 } - \left( R _ { t } - f ^ { \alpha } ( a ) \right) ^ { 2 } \right|
  &= \left| 2 R _ { t } \left( f ^ { \alpha } ( a ) - f ( a ) \right) + f ( a ) ^ { 2 } - f ^ { \alpha } ( a ) ^ { 2 } \right| \nonumber\\
  &\leq 2 \alpha \left| R _ { t } \right| + 2 C \alpha + \alpha ^ { 2 }.
  \label{eq:diseq2}
\end{align}
By Summing up Eq.~\eqref{eq:diseq1} and Eq.~\eqref{eq:diseq2}, we have \begin{eqnarray}
   \sum _ { k = 1 } ^ { t - 1 } \left( \frac { 1 } { 2 } \left[ 4 C \alpha + \alpha ^ { 2 } \right] + \left[ 2 \alpha \left| R _ { k } \right| + 2 C \alpha + \alpha ^ { 2 } \right] \right)
  \leq \alpha \sum _ { k = 1 } ^ { t - 1 } \left( \frac{11}{2} C + 2 \left| R _ { k } \right| \right).
\end{eqnarray}
Recall that $R_k = f + \epsilon_k$, thus $\left| R _ { k } \right| \leq C + \left| \epsilon _ { k } \right|$. According to the Chebyshev Inequality, we have
\begin{eqnarray}
  \mathbb { P } \left( \left| \epsilon _ { k } \right| >  1 / 2 \ln ( 1 / \delta )  \right) \leq \frac{\mathbb{E} [\exp(\epsilon _k)]}{\exp (1 / 2 \ln ( 1 / \delta) )} \leq \delta,
\end{eqnarray}
where we use $\mathbb{E} [\exp(\epsilon _k)] \leq \exp (1 / 2)$. By using a union bound, we have the following equations,
\begin{align}
  \mathbb { P } \left( \exists k s . t . \left| \epsilon _ { k } \right| >  1 / 2 \ln \left( 2 k ^ { 2 } / \delta \right)  \right) &\leq \frac { \delta } { 2 } \sum _ { k=1 } ^ { \infty } \frac { 1 } { k ^ { 2 } } \leq \delta, \\
  \mathbb { P } \left( \forall k  \left| \epsilon _ { k } \right| \leq  1 / 2 \ln \left( 2 k ^ { 2 } / \delta \right)  \right) &\geq 1- \delta.
\end{align}
Then, we can finish our proof as follows,
\begin{align}
  \alpha \sum _ { k = 1 } ^ { t - 1 } \left( \frac{11}{2} C + 2 \left| R _ { k } \right| \right)
  \leq \alpha t (\frac{15}{2} C + 2 |\epsilon _{k_0}|)
  \leq \alpha t (\frac{15}{2} C +  \ln ( 2 t ^ { 2 } / \delta )),
\end{align}
with probability at least $1-\delta$. Here, $|\epsilon_{k_0}| = \max_{1\leq k \leq t-1}  |\epsilon_k|$ and we have $k_0 < t$.
\end{proof}

\textbf{Lemma 4}. \noindent \emph{For all $\delta > 0$ and $\alpha > 0$, if
\begin{eqnarray}
  \mathcal { F } _ { t } = \left\{ f \in \mathcal { F } : \left\| f - \hat { f } _ { t } ^ { L S } \right\| _ { E _ { t } } \leq \sqrt { \beta _ { t } ^ { * } ( \mathcal { F } , \delta , \alpha ) } \right\}
\end{eqnarray}
for all $t \in \mathbb { N }$, then
\begin{eqnarray}
  \mathbb { P } \left( f _ { \theta } \in \bigcap _ { t = 1 } ^ { \infty } \mathcal { F } _ { t } \right) \geq 1 - 2 \delta
\end{eqnarray}
where
$
  \beta _ { t } ^ { * } ( \mathcal { F } , \delta , \alpha )
  : = 8 \ln \left( N \left( \mathcal { F } , \alpha , \| \cdot \| _ { \infty } \right) / \delta \right)
  + 2 \alpha t \left(\frac{15}{2} C +  \ln ( 2 t ^ { 2 } / \delta ) \right)
$
and $N \left( \mathcal { F } , \alpha , \| \cdot \| _ { \infty } \right)$ denotes the $\alpha$ - covering number of $\mathcal{F}$.
}
\begin{proof}
Let $\mathcal { F } ^ { \alpha } \subset \mathcal { F }$ be an $\alpha$-cover of $\mathcal{F}$ in the sup-norm in the sense that for any $f \in \mathcal { F }$ there is an $f ^ { \alpha } \in \mathcal { F } ^ { \alpha }$, such that $\left\| f ^ { \alpha } - f \right\| _ { \infty } \leq \epsilon$. By using a union bound, with probability at least $1-\delta$,  $\forall t \in \mathbb { N } , f \in \mathcal { F } ^ { \alpha }$,
\begin{eqnarray}
  L _ {  t } \left( f ^ { \alpha } \right) - L _ {  t } \left( f _ { \theta } \right)
  \geq \frac { 1 } { 2 } \left\| f ^ { \alpha } - f _ { \theta } \right\| _ {  E _ { t } } - 4 \sigma ^ { 2 } \log \left( \left| \mathcal { F } ^ { \alpha } \right| / \delta \right).
\end{eqnarray}
Therefore, with probability at least $1-\delta$, for all $t\in \mathbb{N}$ and $f\in \mathcal{F}$:
\begin{align}
   L _ {  t } ( f ) - L _ {  t } \left( f _ { \theta } \right) \geq& \frac { 1 } { 2 } \left\| f - f _ { \theta } \right\| _ { 2 , E _ { t } } ^ { 2 } - 4 \sigma ^ { 2 } \log \left( \left| \mathcal { F } ^ { \alpha } \right| / \delta \right) \nonumber\\
   & +\underbrace { \min _ { f ^ { \alpha } \in \mathcal { F } ^ { \alpha } } \left\{ \frac { \left\| f ^ { \alpha } - f _ { \theta } \right\| _ { E _ { t } } ^ { 2 } } { 2 }  - \frac {  \left\| f - f _ { \theta } \right\| _ { E _ { t } } ^ { 2 } } { 2 } + L _ {  t } ( f ) - L _ {  t } \left( f ^ { \alpha } \right) \right\} } _ { \mbox { Discretization Error } }.
\end{align}
As $L_t(\hat{f}  ^{LS}) -  L_t(f_{\theta} )\leq 0$ by definition, using the discretization error bound, we find that with the probability at least $1-2\delta$,
\begin{eqnarray}
  \frac { 1 } { 2 } \left\| \hat { f } _ { t } ^ { \mathrm { LS } } - f _ { \theta } \right\| _ { E _ { t } } ^ { 2 } \leq 4 \sigma ^ { 2 } \log \left( \left| \mathcal { F } ^ { \alpha } \right| / \delta \right) + \alpha \eta _ { t },
\end{eqnarray}
where $\eta_t = t \left( \frac { 15 } { 2 } C + \ln \left( 2 t ^ { 2 } / \delta \right) \right)$. Taking the infimum over the size of $\alpha$ covers, we have:
\begin{eqnarray}
  \left\| \hat { f } _ { t } ^ { L S } - f _ { \theta } \right\| _ {  E _ { t } }
  \leq \sqrt { 8 \sigma ^ { 2 } \log \left( N \left( \mathcal { F } , \alpha , \| \cdot \| _ { \infty } \right) / \delta \right) + 2 \alpha \eta _ { t } }
  \stackrel { \mathrm { def } } { = }& \sqrt { \beta _ { t } ^ { * } ( \mathcal { F } , \delta , \alpha ) }.
\end{eqnarray}
\end{proof}

\subsection{Proof of Lemma 1}
The proof of Lemma 1 is similar with Proposition 9 in~\cite{russo2014learning}. The main difference is that the variable in our method is set of arms (i.e., items) $s$ instead of a single arm $a$. Following~\cite{russo2014learning}, we first introduce the following definitions.

\textbf{Definition 1.} \noindent \emph{An action $s \in \mathcal{S}$ is \emph{$\epsilon $ - dependent} on actions $\left\{ s _ { 1 } , \dots , s _ { n } \right\} \subseteq \mathcal { S}$ with respect to $\mathcal{F}$ if any pair of functions $f , \tilde { f } \in \mathcal { F }$ satisfying $\sqrt { \sum _ { i = 1 } ^ { n } \left( f \left( s _ { i } \right) - \tilde { f } \left( s _ { i } \right) \right) ^ { 2 } } \leq \epsilon$ also satisfies $|f ( s ) - \tilde { f } ( s )| \leq \epsilon$. On the other hand, $s$ is \emph{$\epsilon$ - independent } of $\left\{ s _ { 1 } , \dots , s _ { n } \right\} $ with respect to $\mathcal{F}$ if $s$ is not $\epsilon$ - dependent on $\left\{ s _ { 1 } , \dots , s _ { n } \right\} $.
}

\textbf{Definition 2.} \noindent \emph{The $\epsilon$ - margin dimension $\operatorname { dim } _ { E } ( \mathcal { F } , \epsilon )$ is the length $d$ of the longest sequence of elements in $\mathcal{S}$ such that, for some $\epsilon ^ { \prime } \geq \epsilon$, every element is $\epsilon^{'}$ - independent of its predecessors.
}

\textbf{Confidence interval width.} \noindent $w _ { \mathcal { F } } ( a ) : = \sup _ { f \in \mathcal { F } } f ( s ) - \inf _ { f \in \mathcal { F } } f ( s )$.

In addition, we also introduce the following propositions, which are modified from~\cite{russo2014learning} to fit our definitions. The proofs of these propositions are similar with that described in~\cite{russo2014learning}.

\textbf{Proposition 1.} \noindent \emph{For all $T \in \mathbb{N}$, if $\inf _ { \rho \in \mathcal { F } _ { \tau } } f _ { \rho } ( s ) \leq f _ { \theta } ( s ) \leq \sup _ { \rho \in \mathcal { F } _ { \tau } } f _ { \rho } ( s )$ for all $\tau \in \mathbb { N }$ and $s \in \mathcal { S }$ with probability at least $1- \frac{1}{T}$ then
\begin{eqnarray}
  Regret\left( T , \pi ^ { \mathrm { PS } } \right) \leq C + \mathbb { E } \sum _ { t = 1 } ^ { T } w _ { \mathcal { F } _ { t } } \left( S _ { t } \right).
\end{eqnarray}
}

\textbf{Proposition 2.} \noindent \emph{If $\left( \beta _ { t } \geq 0 | t \in \mathbb { N } \right)$ is a nondecreasing sequence and $\mathcal { F } _ { t } : = \{ f \in \mathcal { F } : \| f -\hat { f } _ { t } ^ { L S } \| _ { 2 , E _ { t } } \leq \sqrt { \beta _ { t } } \}$ then for all $T \in \mathbb{N} $ and $\epsilon > 0$
\begin{eqnarray}
  \sum _ { t = 1 } ^ { T } \mathbf { 1 } \left( w _ { f _ { t } } \left( S _ { t } \right) > \epsilon \right) \leq \left( \frac { 4 \beta _ { T } } { \epsilon ^ { 2 } } + 1 \right) \operatorname { dim } _ { E } ( \mathcal { F } , \epsilon ).
\end{eqnarray}
}

\textbf{Proposition 3.} \noindent \emph{If $\left( \beta _ { t } \geq 0 | t \in \mathbb { N } \right)$ is a nondecreasing sequence and $\mathcal { F } _ { t } : = \{ f \in \mathcal { F } : \| f -\hat { f } _ { t } ^ { L S } \| _ { 2 , E _ { t } } \leq \sqrt { \beta _ { t } } \}$ then for all $T \in \mathbb{N} $
\begin{eqnarray}
  \sum _ { t = 1 } ^ { T } w _ { \mathcal { F } _ { t } } \left( S _ { t } \right) \leq 1 + \operatorname { dim } _ { E } \left( \mathcal { F } , T ^ { - 1 } \right) C
  + 4 \sqrt { \operatorname { dim } _ { E } \left( \mathcal { F } , T ^ { - 1 } \right) \beta _ { T } T }.
\end{eqnarray}
}

Then, we can finish the proof of Lemma 1, by putting these propositions and definitions together.

\subsection{Proof of Lemma 2}
In order to give a bound for the $\epsilon$ - margin dimension, we first introduce the following equivalent definition which provides a mathematical form for the $\epsilon$ - margin dimension.

\textbf{Equivalent definition of $\epsilon$ - margin dimension.} \noindent \emph{The $\epsilon$ - margin dimension of a class of functions $\mathcal{F}$ is the length of the longest sequence $s_1, ..., s_t$ such that for some $\epsilon ' \geq \epsilon$
\begin{eqnarray}
\nonumber 
  w _ { k } &: =& \sup  \{ \left| \left( f _ { \rho _ { 1 } } - f _ { \rho _ { 1 } } \right) \left( s _ { k } \right) \right| \\
\nonumber
  && : \sqrt { \sum _ { i = 1 } ^ { k - 1 } \left( f _ { \rho _ { 1 } } - f _ { \rho _ { 2 } } \right) ^ { 2 } \left( s _ { i } \right) } \leq \epsilon ^ { \prime } \rho _ { 1 } , \rho _ { 2 } \in \Theta\}  \\
  &>& \epsilon ^ { \prime },
\end{eqnarray}
for each $k \leq t $.}


Let $\rho = \theta - \theta ^{\star}$, $\phi _ { k } = \phi \left( s _ { k } \right) = \sum_{i \in s_k} x_i $, $\Phi _ { k } = \sum _ { i = 1 } ^ { k - 1 } \phi _ { i } \phi _ { i } ^ { T }$. Then, we have $\sum _ { i = 1 } ^ { k - 1 } \left( f _ { \theta } - f _ { \theta ^ { \star } } \right) ^ { 2 } \left( s _ { i } \right) \leq h^2 \rho ^ { T } \Phi _ { k } \rho $, and $\|\rho \|_2 \leq 2S$. The proof follows by bounding the number of times $w_{k} \geq \epsilon ^ { \prime }$ can occur. We finish the proof of Lemma 2 by the following three steps.

\textbf{Step 1:} If $w _ { k } \geq \epsilon ^ { \prime }$, then $\phi _ { k } ^ { T } V _ { k } ^ { - 1 } \phi _ { k } \geq \frac { 1 } { 2 }$, where $V _ { k } : = \Phi _ { k } + \lambda I$ and $\lambda = \left( \frac { \epsilon ^ { \prime } } { 2 S h} \right) ^ { 2 }$.
\begin{proof}
By definition,
\begin{eqnarray}
\nonumber 
  \epsilon ^ { \prime } &\leq& \quad w _ { k } \\
\nonumber
  &\leq& \max \left\{ |h \rho ^ { T } \phi _ { k }| : h^2 \rho ^ { T } \Phi _ { k } \rho \leq \left( \epsilon ^ { \prime } \right) ^ { 2 } , \rho ^ { T } I \rho \leq ( 2 S ) ^ { 2 } \right\} \\
\nonumber
  &\leq & \max \left\{ |h \rho ^ { T } \phi _ { k }| :h^2 \rho ^ { T } V _ { k } \rho  \leq 2 \left( \epsilon ^ { \prime } \right) ^ { 2 } \right\} \\
  & \leq &\sqrt { 2 \left( \epsilon ^ { \prime } \right) ^ { 2 } } \left\| \phi _ { k } \right\| _ { V _ { k } ^ { - 1 } }.
\end{eqnarray}
Because of
\begin{eqnarray}
  h ^ { 2 } \rho ^ { T } \left(\Phi _ { k } + \lambda I \right)\rho _ { k } &\leq& \left( \epsilon ^ { \prime } \right) ^ { 2 } + \lambda h^2 (2S)^2 \\
  \left| h \rho ^ { T } \phi _ { k } \right| &\leq & \left| h \right| \|\rho \|_{V_k} \| \phi _k \|_{V_k^{-1}}
\end{eqnarray}
Then, we get $\left\| \phi _ { k } \right\| _ { V _ { k } } ^ { 2 } \geq 1 / 2$.
\end{proof}

\textbf{Step 2:} If $w _ { i } \geq \epsilon ^ { \prime }$ for each $i < k$ then $\operatorname { det } V _ { k } \geq \lambda ^ { d } \left( \frac { 3 } { 2 } \right) ^ { k - 1 }$ and $\operatorname { det } V _ { k } \leq \left( \frac { \gamma ^ { 2 } ( k - 1 ) } { d } + \lambda \right) ^ { d }$.
\begin{proof}
We have $V _ { k } = V _ { k - 1 } + \phi _ { k } \phi _ { k } ^ { T }$, with Matrix Determinant Lemma $\operatorname { det } \left( \mathbf { A } + \mathbf { u v } ^ { \top } \right) = \left( 1 + \mathbf { v } ^ { \top } \mathbf { A } ^ { - 1 } \mathbf { u } \right) \operatorname { det } ( \mathbf { A } )$. Then, we have
\begin{eqnarray}
\nonumber 
  \operatorname { det } V _ { k } &=& \operatorname { det } V _ { k - 1 } \left( 1 + \phi _ { t } ^ { T } V _ { k-1    } ^ { - 1 } \phi _ { t } \right) \\
\nonumber
  &\geq& \operatorname { det } V _ { k - 1 } \left( \frac { 3 } { 2 } \right) \\
\nonumber
  &\geq& \ldots \\
\nonumber
  &\geq& \operatorname { det } [ \lambda I ] \left( \frac { 3 } { 2 } \right) ^ { k - 1 } \\
  &=& \lambda ^ { d } \left( \frac { 3 } { 2 } \right) ^ { k - 1 }.
\end{eqnarray}
In addition, we know that $\operatorname { det } V _ { k }$ is the product of the eigenvalues of $V_k$, whereas $trace[V_k]$ is the sum, and $\operatorname { det } V _ { k }$ is maximized when all eigenvalues are equal. That is
\begin{eqnarray}
\nonumber 
  \operatorname { det } V _ { k } &\leq& \left( \frac { \operatorname { trace } \left[ V _ { k } \right] } { d } \right) ^ { d } \\
  &\leq &  \left( \frac { \gamma ^ { 2 } ( k - 1 ) } { d } + \lambda \right) ^ { d }.
\end{eqnarray}
\end{proof}

\textbf{Step 3:} Let $\alpha _ { 0 } = \left( \frac { \gamma ^ { 2 } } { \lambda } \right) = \left( \frac { 2 S h \gamma } { \epsilon ^ { \prime } } \right) ^ { 2 }$, from step 2 we get $\left( \frac { 3 } { 2 } \right) ^ { \frac { k - 1 } { d } } \leq \alpha _ { 0 } \left[ \frac { k - 1 } { d } \right] + 1 $. Then, we define
$
  B ( x , \alpha ) = \max \left\{ B : ( 1 + x ) ^ { B } \leq \alpha B + 1 \right\}.
$
The longest length of $\{s_t\}$, such that $w _ { k } > \epsilon ^ { \prime }, \forall k < k_{\max}$ is  $k_{\max} \leq d B \left( 1 / 2 , \alpha _ { 0 } \right) + 1$. As $\ln (1+x) \leq x, \forall x \geq 0$, we can get $B \ln \{ 1 + x \}  \leq \ln \{\alpha \} + \ln B$. Let $y =  \frac{x}{1+x} B$, with inequality $\frac{x}{1+x} \leq \ln (1+x)$ and $\ln x \leq \frac{x}{e}$, we have
\begin{eqnarray}
\nonumber 
  y &\leq& \ln \{  \alpha \} + \ln \frac { 1 + x } { x } + \ln y \\
  &\leq& \ln \{  \alpha \} + \ln \frac { 1 + x } { x } + \frac { y } { e }.
\end{eqnarray}
That is
\begin{eqnarray}
  y &\leq& \frac { e } { e - 1 } \left( \ln \{ \alpha \} + \ln \frac { 1 + x } { x } \right), \\
  B ( x , \alpha ) &\leq& \frac { 1 + x } { x } \frac { e } { e - 1 } \left( \ln \{ \alpha \} + \ln \frac { 1 + x } { x } \right)
\end{eqnarray}
Let $\alpha = \alpha_0$ and $x=\frac{1}{2}$, we can get
\begin{eqnarray}
  k_{\max} = \operatorname { dim } _ { E } ( \mathcal { F } , \epsilon )
  \leq 3 d \frac { e } { e - 1 } \ln \left\{ 3 \left( \frac { 2 S h \gamma } { \epsilon } \right) ^ { 2 } \right\} + 1.
\end{eqnarray}

\end{document}